\documentclass[letterpaper, 10 pt, conference]{ieeeconf}  

\IEEEoverridecommandlockouts                              

\usepackage{graphics} 
\usepackage{epsfig} 
\usepackage{amsmath} 
\usepackage{amssymb}  
\usepackage{color}
\usepackage{graphicx}
\usepackage{cite}
\usepackage[capitalize]{cleveref}
\newtheorem{remark}{Remark}
\newtheorem{lemma}{Lemma}
\usepackage{siunitx}
\usepackage{soul}
\usepackage{subfigure}
\newcommand{\Ge}{\hat{G}}
\newcommand{\Gie}{\Ge^{-1}}
\newcommand{\krn}{\hat{k}}

\newcommand{\fft}{\mathcal{F}}
\newcommand{\mat}[1]{\begin{bmatrix}
		#1
\end{bmatrix}}

\newcommand{\diag}{\mathrm{diag}}

\DeclareMathOperator{\imag}{Im}
\DeclareMathOperator{\real}{Re}
\DeclareMathOperator{\cov}{cov}

\title{\LARGE \bf
Iterative Machine Learning for Precision Trajectory Tracking with Series Elastic Actuators
}

\author{Nathan Banka$^{1}$, W. Tony Piaskowy$^{2}$, Joseph Garbini$^{3}$, and Santosh Devasia$^{4}$ 
\thanks{*This research was conducted at the Ultra-Precision Controls Lab at the University of Washington.}
\thanks{$^{1}$Nathan Banka, Post-doctoral Research Associate, Department of 
Mechanical Engineering, University of Washington.}%
\thanks{$^{2}$W. Tony Piaskowy, PhD M.E. student, University of Washington.}%
\thanks{$^{3}$Joseph Garbini is a Professor in the Department of Mechanical Engineering at the University of Washington.}%
\thanks{$^{4}$Santosh Devasia is a Professor in the Department of Mechanical Engineering at the University of Washington.}%
}

\begin{document}

\maketitle
\thispagestyle{empty}
\pagestyle{empty}

\begin{abstract}
When robots operate in unknown environments small errors in postions can lead to large variations in the contact forces, especially with typical high-impedance designs. This can potentially damage the surroundings and/or the robot. Series elastic actuators (SEAs) are a popular way to reduce the output impedance of a robotic arm to improve 
control authority over the force exerted on the environment. However this increased control over forces with lower impedance comes at the 
cost of lower positioning precision and bandwidth. This article examines the use of an iteratively-learned feedforward command to improve position tracking when using SEAs. Over each iteration, the output 
responses of the system to the quantized inputs are used to estimate  
a linearized local system models. These 
estimated models are obtained using a complex-valued Gaussian Process Regression 
(cGPR) technique and then, used  to generate a new feedforward input 
command based on the previous iteration's error. This article 
illustrates this iterative machine learning (IML) technique for a two degree of freedom (2-DOF) robotic arm,  and demonstrates successful convergence of the IML approach to reduce the  tracking error.

\end{abstract}

\section{Introduction}
Robots are now operated in unstructured environments where good control over the 
applied force is needed, especially if the robot needs to interact with humans 
or is operating in environments where large forces can lead to damage.  Fine 
robot positioning with high bandwidth has been achieved traditionally using high 
impedance (i.e., very stiff) designs. However, such high impedance leads to poor 
force control due to large errors in force with small position deviations 
\cite{Pratt1995}. An increasingly popular solution for improved force control is 
the use of series elastic actuators (SEAs), which were introduced by Pratt and 
Williamson \cite{Pratt1995}, with further control analysis done by Robinson 
\cite{Robinson2000}. SEAs introduce an elastic element in series between the 
motor and its load to lower output impedance. Additionally,   the output force 
can be directly estimated by  measuring the deformation of the elastic element 
\cite{Pratt1995}, which can be used in safety-critical application, e.g., to 
stop the robot action if the forces are excessive.    The original elastic 
actuator concept has inspired many other designs and their use in human-robot 
interaction applications, such as rehabilitation 
\cite{Kim2014,Nagarajan2016,Li2016Adaptive}. While adequate force control is 
achieved with SEAs, this is at the cost of lower performance in terms of 
positioning precision and bandwidth,  and distrubance rejection \cite{Paine2014, 
Chawda2017} due to lowered impedance, friction nonlinearities, and backlash 
\cite{Eppinger1988, Robinson2000}.
Recent control methods attempt to remedy these problems by modeling the higher 
order system dynamics associated with the SEAs, active system identification, 
and disturbance estimation \cite{Muller2017, Sariyildiz2017, Li2016}. 
Nevertheless, feedback methods tend to have performance limitations due to the 
nonminimum-phase behavior with the elastic elements. 

Model-based methods, such as inversion to find feedforward inputs to augment the 
feedback, can improve  performance of robots with elastic joints, 
e.g.,~\cite{Lucibella_98}. Moreover, the effects of modeling uncertainty can be 
corrected using Iterative Learning Control (ILC) approaches, which 
can improve the tracking performance for repeated tasks \cite{McDaid12, Tien05, 
Butcher2010, jayati00, Tsai2013, Mishra2009, Havlicsek1999, Atkeson86}. 
Nevertheless, if the initial modeling errors are large, then the iterative 
approach could diverge \cite{Moore07_review}. This has motivated approaches that 
use the data obtained during the iteration process without an explicit model to  
improve convergence~\cite{Zou_ACC_08}. Additionally,  kernel-based Gaussian 
process regression has been used to update the models using iteration data 
\cite{Devasia17IML} and such iterative machine learning (IML) approaches have 
been experimentally evaluated  in \cite{Banka17IML}.  The extension of such IML 
methods for finding the feedforward input to improve the precision when using 
SEAs is investigated in this article. 

A challenge with  modeling multi-linked robotic systems is that the system 
dynamics are nonlinear.  To develop iterative corrections, the trajectory 
tracking control approach in this article finds localized linear system 
representations of the nonlinear system.   
The input-output data is then 
used to derive locally (linearized) models. For sufficiently-slow changes in the 
trajectory, the  localized models can be used to develop iterative input 
corrections. The local models are obtained in the frequency domain using using 
complex-valued Gaussian Process Regression (cGPR) for fitting the measured 
input-output data. An advantage of the cGPR approach is that it also provides 
potential uncertainties in the models, which can be used to design the gains in 
the  iteration law to promote convergence.  The iteration approach allows large 
changes in input updates in frequency regions where the model uncertaintly is 
small. No changes or small changes in the in the input updates are used if the 
predicted model uncertainty is large\cite{Blanken17,Liu15,Norrlof_02}.

\section{Problem Formulation and Solution}
This section introduces the iterative machine learning (IML) procedure, which 
has previously been applied to scalar linear time-invariant (LTI) systems 
\cite{Devasia17IML,Banka17IML} and presents two extensions for applications to: 
(i) linear time-varying (LTV) systems and (ii) multiple-input multiple-output 
(MIMO) systems. The application to LTV systems also allows for control of 
control-affine nonlinear systems that can be approximated as LTV systems 
parameterized about the desired trajectory.

\subsection{Iterative machine learning for scalar LTI systems}
Consider a scalar LTI system, where the input $u$ and output $y$ are related by the 
impulse response $g$ as
\begin{align}
	y(t) = (g \star u)(t), \label{eq:impulse-response}
\end{align}
where $\star$ denotes convolution in the time domain, or equivalently in frequency domain by
\begin{align}
	Y(\omega) = G(\omega) U(\omega),
    \label{eq:model_freq_response}
\end{align}
where $Y(\omega) = \fft(y( \cdot ))$, $U(\omega) = \fft(u( \cdot ))$, and $\fft$ denotes the 
Fourier transform. 
The tracking problem for this system is to find an input 
$u^*$ such that the system output $y$ is equal to the desired output 
$y_d$, i.e., $y(t) = y_d(t)$.  
%
In the iterative learning control method, 
the input signal is iteratively 
corrected; given a previous input $U_{k-1}(\omega)$ which resulted in an output 
$Y_{k-1}(\omega)$, the correction is
\begin{align}
	U_k(\omega) = U_{k-1}(\omega) + \rho(\omega) \Gie(\omega) 
	(Y_d(\omega)-Y_{k-1}(\omega)),\label{eq:iterative-update-f}
\end{align}
where $\Gie$ is an estimate of the inverse of the transfer function $G$ 
and $\rho$ is the frequency-dependent iteration gain.  It was shown in \cite{Tien05} that 
this sequence converges to the desired input $U^*(\omega)$ at frequency $\omega$ 
if and only if
\begin{align}
	|\Delta_p(\omega)| < \pi/2\label{eq:phase-criterion}\\
	0 < \rho(\omega) < 
	\frac{2\cos\Delta_p(\omega)}{\Delta_m(\omega)}\label{eq:convergence-criterion}
\end{align}
where $\Delta_m$ and $\Delta_p$ are the magnitude and phase of the modeling 
error given by
\begin{align}
	\frac{G^{-1}}{\Gie} = \Delta_m e^{j\Delta_p}.
\end{align}

Data-based modeling approaches can improve the convergence of the iterative 
control by increasing the accuracy of the inverse system model $\Gie(\omega)$ 
used in the iteration law in Eq.~\eqref{eq:iterative-update-f}. In the scalar 
LTI case, estimates of the inverse system model can be computed in the frequency 
domain as~\cite{Zou_ACC_08}
\begin{align}
	\Gie(\omega) = U(\omega)/Y(\omega).\label{eq:model-data}
\end{align}
When the data is noisy or unavailable at some frequencies, such model estimates 
can be improved by using the data to train a complex-valued Gaussian process 
regression (cGPR) model of the form
\begin{align}
	\Gie(\omega) = f(\omega) + \epsilon,
\end{align}
where $f(\omega)$ is the underlying true function and $\epsilon$ is a Gaussian 
noise term. In the Gaussian process framework, function values $f(\omega_1)$ and 
$f(\omega_2)$ are assumed to be correlated according to a kernel function 
$\krn(\omega_1,\omega_2)$ which measures the similarity of the inputs $\omega_1$ 
and $\omega_2$.  Given measured data $\Gie_T \in \mathbb{C}^n$ at frequencies 
$\Omega_T \in \mathbb{R}^n$, the resulting estimated mean $\mathbb{E}$ and 
variance $\mathbb{V}$ of the underlying function $f(\omega)$ are given 
by\cite{Rasmussen_06}
\begin{align}
	\mathbb{E}[f(\omega)] &= K(\omega,\Omega_T) K_{T}^{-1} \Gie_T(\Omega_T)
	\label{eq:Gi-mean},\\
	\mathbb{V}[f(\omega)] &= K(\omega,\omega) - K(\omega,\Omega_T) K_{T}^{-1} 
	K(\Omega_T,\omega),\label{eq:Gi-variance}
\end{align}
where $K(\Omega_1,\Omega_2)$ is the covariance matrix formed by evaluating the 
kernel $\krn(\cdot, \cdot)$ on all pair of entries from $\Omega_1$ and 
$\Omega_2$ and $K_T = K(\Omega_T,\Omega_T) + \sigma_\epsilon^2 I$  is the 
covariance matrix of the noisy training data.

\subsection{Extension to LTV case}
\label{sec:ltv}
When the system is nonlinear, a transfer function approach is not valid in 
general.  In the following, the system is approximated using an LTI model at 
every time instant $t$,  which can vary over time, i.e., 
\cref{eq:model_freq_response} as
\begin{align}
	y(t) &= (g_t \star u)(t)\\
	g_t &= h(t, p_1(t), p_2(t),\dots,p_m(t)).
\end{align}
where the local system transfer function $g_t$ can vary with time, as can the 
potentially time-varying parameters $p_1, p_2, \dots, p_m$, which are assumed to 
be  known or estimated at all times.


\subsubsection{Input update}
The updated input signal at iteration $k$ evaluated at a specific time $t_j$, 
$u_k(t_j)$, can be written as in \cref{eq:iterative-update-f} with the model $ 
{G}_{t_j}$, in the frequency domain as \begin{align}
	U_k(\omega) = U_{k-1}(\omega) + \rho(\omega) {G}_{t_j}^{-1} (\omega) 
	E_{k-1}(\omega),
\end{align}
where $e_k(t) = y_d(t)-y_k(t)$ is the error at iteration $k$ and time $t$, which is 
equivalent to, in the time 
domain, 
\begin{align}
	u_k(t_j) & = u_{k-1}(t_j) + (\rho \star g_{t_j}^{-1}\star 
	e_{k-1})(t_j),\label{eq:ltv-update}
\end{align}
 Note that updating the input at time instant $t_j$ using the convolution requires the error $e_k(t)$ at time instants $t$ before and after the time-instant $t_j$. 

\subsubsection{Model for input update}
In the Gaussian process 
framework, learning the parameterized model means the system is taken to have 
the form
\begin{align}
	\Gie(\mathbf{x}) = f(\mathbf{x}) + \epsilon,
\end{align}
where $\mathbf{x} = [\omega, p_1, p_2,\dots,p_m]$ and the kernel function 
$\krn(\mathbf{x}_1, \mathbf{x}_2)$ depends in some sense on the distance between 
points in $\mathbb{R}^{m+1}$. The training data in this case consists of sampled 
local transfer function estimates $\Gie_T$, which are computed by windowing the 
time-domain data and applying the Fourier transform to each segment, the 
corresponding frequencies $\Omega_T$, and representative parameter values for 
each segment.  Then the locally-linear inverse transfer function $\Gie(\omega)$ 
can be computed for each set of parameters $p_i(t)$ and the iterative update in 
\cref{eq:ltv-update} can be applied. 
The input correction can be computed for each model in frequency domain and 
sampled to retrieve the time index of interest, resulting in the update law
\begin{align}
	u_k(t_j) = u_{k-1}(t_j) + 
	\fft^{-1}\left.\left(\rho\Gie_j\fft(y_d(t)-y_{k-1}(t))\right)\right|_{t = 
		t_j}\label{eq:ltv-update-f}
\end{align}
where $\Gie_j \equiv \Gie(\omega,p_1(t_j),p_2(t_j),\dots,p_m(t_j)) \in 
\mathbb{C}^n$ is the estimated model based on the parameters at time $t_j$ 
evaluated at all frequencies.





\subsection{Extension to multivariate systems}
\label{sec:mimo}
When the system has more than one input, i.e.
\begin{align}
	Y_i(\omega) = \sum_{j=1}^p G_{ij}(\omega) U_j(\omega),
\end{align}
the individual transfer functions $G_i$ cannot be directly estimated from the 
data as in \cref{eq:model-data}. Experimentally it may be viable to sample each 
transfer function by holding all other inputs zero; however, it is preferable to 
develop a learning method that is valid for the general case.  

\subsubsection{Weighted kernel for single input case} Consider the case 
when the number of inputs $p = 1$ and samples of the input $U_1 
\in \mathbb{C}^n$ and corresponding output $Y_1 \in\mathbb{C}^n$ are available 
at frequencies $\Omega\in\mathbb{R}^n$.  Then, similar to \cref{eq:model-data}, 
the transfer function can be estimated as a Gaussian process
\begin{align}
	\Ge_{11}(\omega) = f_{11}(\omega) + \epsilon
\end{align}
with a kernel function $\krn_{11}(\omega_r,\omega_s)$ 
and noisy 
	transfer function samples $\Ge_{11r} = Y_{1r}/U_{1r}$ at frequency 
	$\omega_r$.  Equivalently, multiplying by the  known input 
	$U_1(\omega)$, the transfer function $G_{11}$ can be estimated 
	from 
\begin{align}
	Y_1(\omega) = f'_1(\omega, U_1) + \epsilon  = f_{11}(\omega) U_{1}(\omega) 
	+ \epsilon \label{eq:gp-weighted}
\end{align}
where the factor of $U_1(\omega)$ on the noise term $\epsilon$ has been dropped 
since the noise can be considered as output noise rather than process noise.  
The covariance of the input-weighted kernel $f_1'$ is given by 
	\begin{align}
		\notag\krn_{11}'(\omega_r,\omega_s|U_{1r},U_{1s}) &= 
		\cov(f'_1(\omega_r,U_{1r}),f_1'(\omega_s,U_{1s}))\\
		\notag &= 
		\mathbb{E}[f'_1(\omega_r,U_{1r})\overline{f}_1'(\omega_s,U_{1s})]\\
		\notag &= \mathbb{E}[f_{11}(\omega_r)U_{1r} 
		(\overline{f}_{11}(\omega_s)\overline{U}_{1s})]\\
		\notag &= 
		U_{1r}\mathbb{E}[f_{11}(\omega_r)\overline{f}_{11}(\omega_s)] 
		\overline{U}_{1s}\\
		{}& = U_{1r} \krn_{11}(\omega_r,\omega_s) 
		\overline{U}_{1s}\label{eq:weighted-kernel}.
\end{align}
where the overbar denotes complex conjugation and with an equivalent result in 
the case that $f_{11}$ is not a zero-mean process. The matched subscripts on 
$\omega_r$ and $U_{1r}$ indicate that $U_{1r}$ is an entry of the overall data 
vector $U_1$ with corresponding frequency $\omega_r$ from the data vector 
$\Omega$.

\begin{remark}[Construction of covariance matrix]
	The input-weighted kernel $\krn_{11}'$ in \cref{eq:weighted-kernel} gives 
	rise to a covariance matrix of the form
	\begin{align}
		K'_{11}(\Omega,\Omega|U_1,U_1) = \diag(U_1) K_{11}(\Omega,\Omega) 
		\diag(\overline{U}_1),
	\end{align}
	where $K_{11}$ is the covariance matrix for the unweighted kernel 
	$\krn_{11}$. The input-weighted kernel $\krn'_{11}$ is a valid kernel if 
	and only if the covariance matrix $K'_{11}$ is positive semi-definite for 
	any $\Omega$ and $U$, or, with  the conjugate transpose denoted  by 
	superscript $H$,
	\begin{align}
		\notag 0 &\leq v^H K'_{11}(\Omega,\Omega|U_1,U_1) v\\
		\notag&= v^H \diag(U_1) K_{11}(\Omega,\Omega) \diag(U_1)^H v\\
		\notag&= (\diag(U_1)^H v)^H K_{11}(\Omega,\Omega) (\diag(U_1)^H v)\\
		\notag&= w^H K_{11}(\Omega,\Omega) w
	\end{align}
which is satisfied for any valid unweighted kernel $\krn_{11}$.
\end{remark}
\begin{remark}[Equivalence of weighted kernel]
	Given measured frequency-domain input and output data $U_1$ and $Y_1$ at 
	corresponding frequencies $\Omega$, the predictive mean at test frequencies 
	$\Omega_*$ from the unweighted kernel is given by \cref{eq:Gi-mean} as 
	\begin{align}
		\mathbb{E}[f_{11}(\Omega_*)] = K_*(K+C_\epsilon)^{-1} \tilde{G},
	\end{align}
	where $K = K_{11}(\Omega,\Omega)$, $K_* = K_{11}(\Omega_*,\Omega)$, 
	$C_\epsilon$ is the noise covariance matrix, and the estimated transfer 
	function samples $\tilde{G}_{11}$ are computed by entry-wise division of 
	$Y$ by $U$, i.e. $Y_1 = \diag(U_1) \tilde{G}_{11}$. Similarly, when using 
	the input-weighted kernel $\krn'_{11}$ with test inputs $U_{1*}$ at 
	corresponding frequencies $\Omega_*$, with $D \equiv \diag(U_1)$ and $D_* 
	\equiv \diag(U_{1*})$, and treating the noise as process noise with 
	$C'_\epsilon = DC_\epsilon\overline{D}$, the predictive mean is given by
	\begin{align}
		\notag\mathbb{E}[f'_{11}(\Omega_*,U_*)]&= K'_* (K'+C'_\epsilon)^{-1} 
		Y_1\\
		\notag&=
		D_*K_*\overline{D}
		(D(K+C_\epsilon)\overline{D})^{-1} D \tilde{G}_{11}\\
		\notag&=D_*K_*(K+C_\epsilon)^{-1}\tilde{G}_{11}\\
		{}& = D_* \mathbb{E}[f_{11}(\Omega_*)]
	\end{align}
	so by selecting $D_* = I$, an estimate of the transfer function $G$ is 
	obtained that is equivalent to the estimate from the unweighted kernel 
	$\krn_{11}$.
\end{remark}





\subsubsection{Generalization to multiple input case}
In the case where the number of inputs $p$ is greater than 1, each transfer 
function $G_{ij}$ can be modeled as an independent Gaussian process with kernel 
$\krn_{ij}$, resulting in a Gaussian process model for $Y_i$ as
\begin{align}
	Y_i(\omega) &= f_i'(\omega,U_1,\dots,U_p)= \sum_{j=1}^p f_{ij}(\omega) 
	U_j(\omega) + \epsilon
\end{align}
and the weighted kernel function $\krn'_i$  can be derived using 
	the identity that for independent random vectors $X$, $Y$,
\begin{align}
	\cov(X+Y,X+Y) = \cov(X,X)+\cov(Y,Y),
\end{align}
resulting in
\begin{align}
	\krn'_i(\omega_r,\omega_s|\{U_{jr},U_{js}\}_{j=1}^p) = 
	\sum_{j=1}^p 
	U_{jr}\krn_{ij}(\omega_r,\omega_s)\overline{U}_{js}\label{eq:mimo-gpr}
\end{align}
and the corresponding covariance matrix \begin{align}
	K'_i = \sum_{j=1}^p\diag(U_j) K_{ij} \diag(\overline{U}_j).
\end{align}
Then an estimate of an individual transfer function $\Ge_{ij}$ can be 
obtained by generating a Gaussian process prediction with $U_{j*} = 
\mathbf{1}$ and all other inputs identically zero.
In the multiple-output case, \cref{eq:mimo-gpr} represents the kernel for 
estimating one line of the transfer matrix, i.e. each output corresponds to a 
multiple-input Gaussian process model with a kernel as in \cref{eq:mimo-gpr}.

\subsubsection{Convergence conditions} 
In the case of square MIMO systems, conditions for  local convergence are given by 
the following lemmas.
\begin{lemma}
	\label{lem:convergence}
	 For a fixed transfer matrix $G(\omega) \in \mathbb{C}^{N\times N}$, where $N$ 
	is a positive integer,  the sequence in \cref{eq:iterative-update-f} 
	converges to $U^*(\omega) \in \mathbb{C}^N$ at frequency 
	$\omega$ if and only if all eigenvalues of 
	$(I-\rho(\omega)\Gie(\omega)G(\omega))$ have magnitude less than one.
\end{lemma}
\begin{proof}
	From \cref{eq:iterative-update-f},
	\begin{multline}
		U_{k+1}-U_k = U_k +\rho\Gie(Y_d-Y_k)\\ - U_{k-1} - 
		\rho\Gie(Y_d-Y_{k-1}),
	\end{multline}
	where the argument $\omega$ has been dropped for clarity. Then, cancelling 
	the repeated factors of $Y_d$ and using the definition of the transfer 
	matrix $G$,
	\begin{align}
		U_{k+1} - U_k &= U_k - U_{k-1} -\rho\Gie G(U_k - U_{k-1})\\
		{}&= (I-\rho\Gie G)(U_k-U_{k-1})\\
		{}&= (I-\rho\Gie G)^k(U_1-U_0).
	\end{align}
	The sequence $U_k$ converges as $k \rightarrow \infty$  if and only if
	\begin{align}
		\lim_{k\rightarrow\infty} (I-\rho\Gie G)^k = 0,
	\end{align}
	or, equivalently, if all eigenvalues of $(I-\rho\Gie G)$ have magnitude 
	less than one.
\end{proof}

\begin{lemma}[Convergence for square MIMO systems]
	\label{lem:rho-bounds}
	Let $G(\omega) \in \mathbb{C}^{N\times N}$ be the fixed 
	transfer matrix of a system with $N$ inputs and $N$ outputs evaluated at 
	frequency $\omega$ and let $\Gie(\omega)$ be an estimate of its inverse. 
	Let the modeling error be described as $\Delta(\omega) \equiv \Gie G$, with
	\begin{align}
		[\Delta]_{ij} = \Delta_{mij}e^{j\Delta_{pij}}
	\end{align}
	where the diagonal terms are nonzero, i.e., $\Delta_{mii} \ne 0 $ for all $i$. 
	Then the sequence
	\begin{align}
		U_{k}(\omega) = U_{k-1}(\omega) + \rho(\omega) 
		\Gie(\omega)(Y_d(\omega) - Y_{k-1}(\omega))
	\end{align}
	converges to the desired input $U^*(\omega)$ at frequency $\omega$ if the 
	iteration gain $\rho(\omega) = \diag(\rho_1(\omega),\dots,\rho_N(\omega))$ 
	satisfies the conditions
	\begin{align}
		0 < \rho_i(\omega) &< 2\frac{\Delta_{mii}\cos\Delta_{pii} - 
		\sum_{j\neq i} \Delta_{mij}}{\Delta_{mii}^2 - \left(\sum_{j\neq 
	i}\Delta_{mij}\right)^2}\label{eq:mimo-rho-bound}\\
	\Delta_{mii}\cos\Delta_{pii} &> \sum_{j\neq 
	i}\Delta_{mij}\label{eq:mimo-phase-bound}
	\end{align}
	where the subscript $j$ is an index and not $\sqrt{-1}$. 
\end{lemma}
\begin{proof}
	By \cref{lem:convergence}, proving convergence amounts to proving the 
	magnitudes of the eigenvalues of $(I-\rho\Gie G)$ are less than 1.  By the 
	Gershgorin disc theorem, the eigenvalues lie in the union of the discs with 
	centers $C_i = 1-\rho_i\Delta_{ii}$ and radii $R_i = \sum_{j\neq 
	i}|\rho_{i}\Delta_{ij}| = \rho_i\sum_{j\neq i}\Delta_{mij}$ for $i \leq N$.  
	The maximum distance from the origin to a point in disc $i$ is then given 
	by $|C_i|+R_i$, so disc $i$ is bounded by the unit circle if and only if
	\begin{align}
		1 >{}& |1-\rho_i\Delta_{ii}| + R_i \label{eq:rho-bound-start} \\
		\notag 1-R_i >{}& 
		|1-\rho_i\Delta_{mii}(\cos\Delta_{pii}+j\sin\Delta_{pii})|. 
	\end{align}
	If the above statement is true, then $(1-R_i) \ge 0 $ and therefore, the square must also hold, i.e., 		
	\begin{align}	
		\notag (1-R_i)^2 >{}& (1-\rho_i\Delta_{mii}\cos\Delta_{pii})^2
		+\rho_i^2\Delta_{mii}^2\sin^2\Delta_{pii}\\
		\notag 1-2R_i+R_i^2 >{}& 1 + \rho_i^2\Delta_{mii}^2 - 
		2\rho_i\Delta_{mii}\cos\Delta_{pii}.
	\end{align}
Finally, after cancelling  the like terms and dividing by the positive gain 
$\rho_i$, this expression can be rearranged as 
	\begin{align}	
	         \rho_i   [ \Delta_{mii}^2 -  (\sum_{j\neq i}\Delta_{mij} )^2 ]  < {} &   	
	      	   2( \Delta_{mii} \cos\Delta_{pii}  - \sum_{j\neq i}\Delta_{mij}  ) . 
		    \label{eq:proof_before_division}
	\end{align}
	Since  $\Delta_{mii} \ne 0 $, the 	
	condition in \cref{eq:mimo-phase-bound}  implies that 
	\begin{align} 
	\notag  \cos\Delta_{pii} > 0 . 
	\end{align}
	Moreover, since $\cos\Delta_{pii}$ is less than one, 	 from the condition in \cref{eq:mimo-phase-bound}, 
	\begin{align} 
	 \Delta_{mii} >   \sum_{j\neq i}\Delta_{mij} \ge  0 , 
	 \label{eq:mimo-phase-bound_no_phase}
	\end{align}
	and squaring the uncertainties on both sides yields 
	\begin{align}
	\Delta_{mii}^2  - (\sum_{j\neq i}\Delta_{mij})^2  > 0 , 
	\end{align}
	which leads to  \cref{eq:mimo-rho-bound} from \cref{eq:proof_before_division}.  
	The use of the phase condition in \cref{eq:mimo-phase-bound},   rather than a condition as in \cref{eq:mimo-phase-bound_no_phase}, 
	ensures that the right hand side of \cref{eq:mimo-rho-bound} is nonzero, which allows the selection of a nonzero iteration gain $\rho_i$. 
	
%
%
%
\end{proof}
\begin{remark}[Specialization to scalar case]
    When the dimension of the transfer matrix $N = 1$, or when the error in the 
    off-diagonal terms is zero, \cref{eq:mimo-rho-bound,eq:mimo-phase-bound} are 
    equivalent to the expressions given for scalar systems in 
    \cref{eq:phase-criterion,eq:convergence-criterion}.
\end{remark}

\begin{lemma}[Convergence with bounded uncertainty]
	\label{lem:bounded-error-convergence}
	When the errors in the real and imaginary components of the estimated model 
	$\Ge$ are bounded as
	\begin{align}
		|\hat{a}_{ij} - a_{ij}| \leq \Delta_{aij}
		\label{def_uncertainty_a}
		\\
		|\hat{b}_{ij} - b_{ij}| \leq \Delta_{bij}
		\label{def_uncertainty_b}
	\end{align}
	with
	\begin{align}
		G_{ij} &= a_{ij} + jb_{ij},\label{eq:uncertainty-real}\\
		\Ge_{ij} &= \hat{a}_{ij} + j\hat{b}_{ij}\label{eq:uncertainty-imag},
	\end{align}
	the condition in \cref{eq:mimo-rho-bound} is satisfied if
	\begin{align}
		0 < \rho_i < 
		2\frac{1-\real\hat{r}_{abs,i}\Delta_{ai}-\imag\hat{r}_{abs,i}\Delta_{bi}- 
			|\hat{r}_i|D_i}{|\hat{r}_i|^2|\hat{c}_{abs,i}+ \Delta_{ai} + 
			j\Delta_{bi}|^2\label{eq:mimo-gpr-rho}
	}
	\end{align}
	where $\hat{r}_i$ is the $i$th row of the inverse transfer matrix $\Gie$, 
	$\hat{c}_i$ is the $i$th  column of the forward transfer matrix $\Ge$, 
	$\Delta_{ai}$ and $\Delta_{bi}$ denote the $i$th column of the real and 
	imaginary uncertainty matrices $\Delta_a$ and $\Delta_b$ defined in \cref{def_uncertainty_a}, \cref{def_uncertainty_b}, 
	\begin{align}
		[\hat{c}_{abs,i}]_j &= |\real\hat{c}_{ij}| + j|\imag\hat{c}_{ij}|,\\
		D_i &= \sum_{j\neq i} |\Delta_{aj}+j\Delta{bj}|, 
	\end{align}
	\begin{align}
		[\real\hat{r}_{abs,i} ]_j  & = |\real\hat{r}_{ij}|  \\
	        [\imag\hat{r}_{abs,i} ]_j & = |\imag\hat{r}_{ij}| . 
   	\end{align}
\end{lemma}
\begin{proof}
	The lemma amounts to finding a lower bound on the right-hand-side of 
	\cref{eq:mimo-rho-bound} under the uncertainty $\Delta_a+j\Delta_b$. It is 
	convenient to first develop bounds on the terms in the fraction.  Note that 
	$\Delta_{ij} = \hat{r}_i c_j$, where $c_j$ is the $j$th column of the 
	forward transfer matrix $G$.  Then the magnitude $\Delta_{mij}$ is
	\begin{align}
		\Delta_{mij}&=|\hat{r}_i c_j| = |\hat{r}_i (\hat{c}_j + e_j)|
		= |\delta_{ij} + \hat{r}_i e_j|
	\end{align}
	where the real and imaginary parts of the error term $e_j$ is bounded by 
	\cref{eq:uncertainty-real,eq:uncertainty-imag}. Then, using the 
	Cauchy-Schwarz inequality, the magnitude can be bounded as
	\begin{align}
		\Delta_{mij} \leq
		\begin{cases}
			|\hat{r}_i| |\Delta_{aj}+j\Delta_{bj}|, & i \neq j\\
			|\hat{r}_i| |\hat{c}_{abs,i}+\Delta_{ai}+j\Delta_{bi}|, & i = j
		\end{cases}
	\end{align}
	where in the case of $i = j$ the second term is the maximum magnitude of 
	$|\hat{c}_i|$ due to the bounds in 
	\cref{eq:uncertainty-real,eq:uncertainty-imag}.

	Similarly,
	\begin{align}
		\Delta_{mii}\cos\Delta_{pii} &= \real(\hat{r}_i c_i) = 
		\real(\delta_{ii} + \hat{r}_ie_i) \\&= 1+ \real(\hat{r}_i)\real(e_i) - 
		\imag(\hat{r}_i)\imag(e_i), 
	\end{align}
	which is minimized when every entry of  $e_i$ has the same (maximum) magnitude as  $\Delta_{ai} + j \Delta_{bi}$ and every term in the 
	product is negative, resulting in
	\begin{align}
		\Delta_{mii}\cos\Delta_{pii} \geq 1 - \real\hat{r}_{abs,i}\Delta_{ai} 
		- \imag\hat{r}_{abs,i}\Delta_{bi}.
	\end{align}
	Combining these bounds on the terms in \cref{eq:mimo-rho-bound} leads to 
	the expression in \cref{eq:mimo-gpr-rho}.
\end{proof}


\begin{remark}[Local convergence of time-varying model]
	In the preceding 
	\cref{lem:convergence,lem:rho-bounds,lem:bounded-error-convergence}, the iteration-law was designed assuming a 
	time-invariant system. For the time-varying case, as in \cref{sec:ltv}, the 
	iteration law designed using the fixed model is anticipated to converge 
	provided the time-variations are sufficiently slow. The convergence 
	conditions for such variations could be developed using Picard-type 
	arguments, e.g., as studied in   \cite{Altin17} for the nonlinear case.  
\end{remark}

\section{Experiments}
\subsection{Experimental system}
The experimental system, used to evaluate the proposed IML approach, was a 2-DOF 
robotic arm, illustrated in Figures \ref{fig:ArmDiagram} and \ref{fig:ArmPhoto}. 
The  arm was constructed out of HEBI X5-4 series elastic actuators, black 
1.25''$\oslash$ PVC Pipe, aluminum brackets, and a steel plate for mass, as 
shown in Figure \ref{fig:ArmPhoto}. Each  X5-4 actuator contains: a 
microcontroller, motor drive, Ethernet communication/daisy chain, encoder, 
series elastic element, and status LED. The X5-4 is capable of 
\SI{4}{\newton\meter} continuous torque (\SI{7}{\newton\meter} peak), 32 RPM 
max, $\pm$ 4 turn absolute encoder with 0.005$^{\circ}$ resolution, and contains 
a torsional spring in series between motor and load with a stiffness of about 
\SI{70}{\newton\meter\per\radian}.  Each actuator weighs approximately 
\SI{340}{\g}.

MATLAB was used to communicate over Ethernet with the X5's on-board microcontrollers. Data was logged locally on each actuator, and subsequently retrieved and processed via MATLAB interface.

\begin{figure}
\centering
\includegraphics[width=0.7\columnwidth]{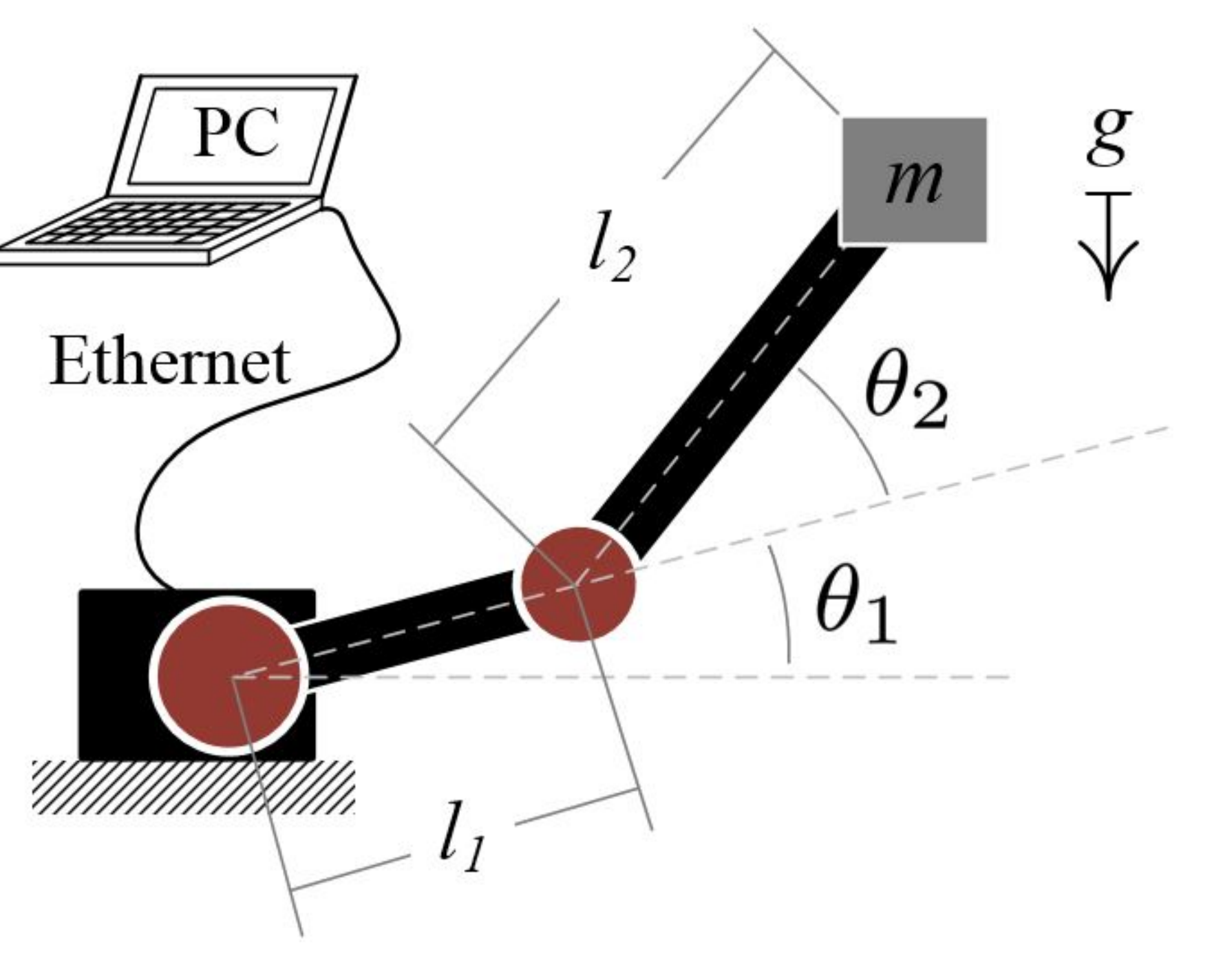}
\caption{System diagram. 2-DOF Robotic Arm, driven angles $\theta_1$ and 
$\theta_2$, with link lengths of $l_1$ = 6.25'' and $l_2$ = 10.5'', supporting a 
mass, $m$ = \SI{275}{\gram}, at the end of link 2. Actuator communication to PC 
over Ethernet.}
\label{fig:ArmDiagram}
\end{figure}

\begin{figure}
\centering
\includegraphics[width=0.7\columnwidth]{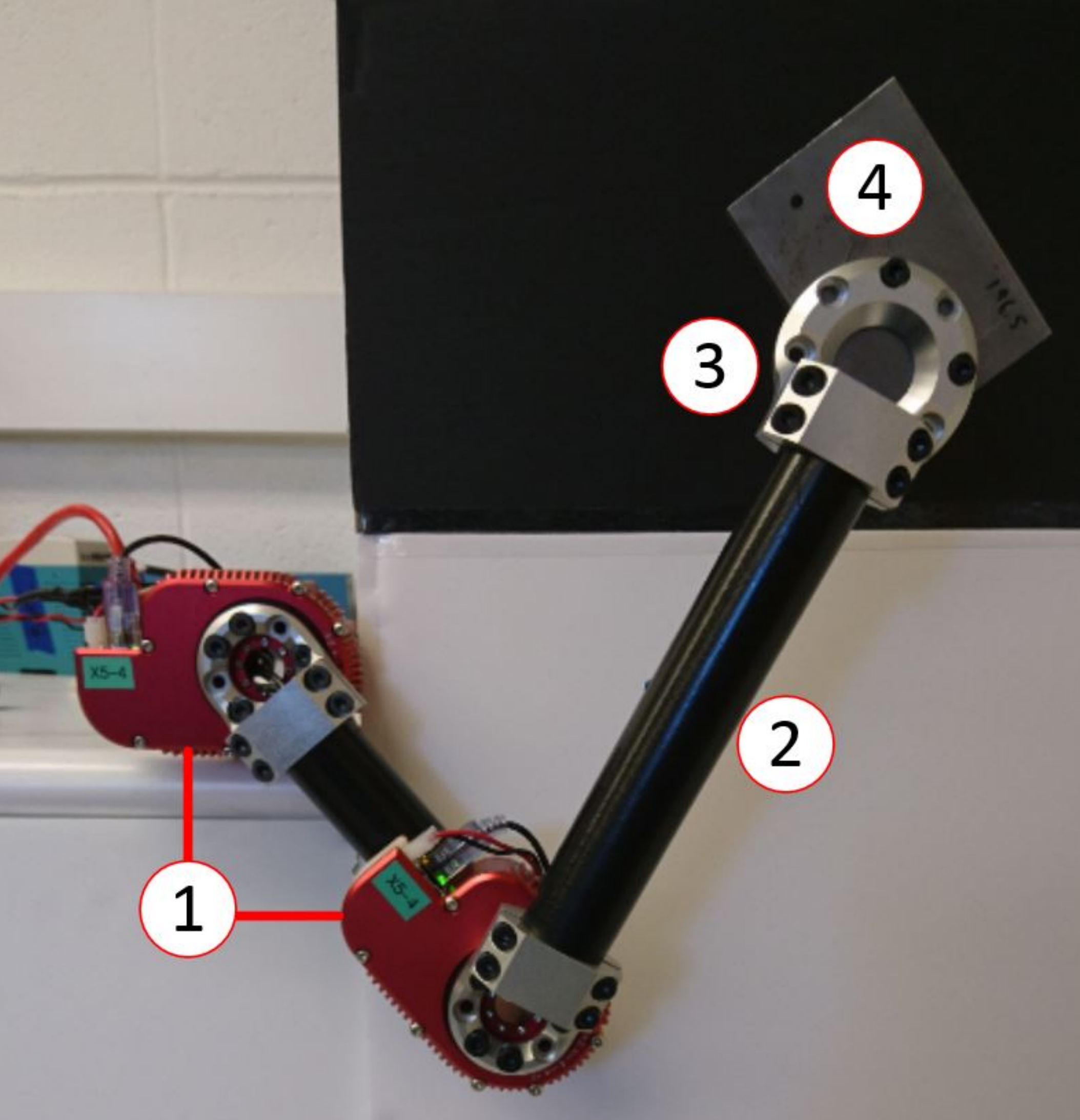}
\caption{Experimental Setup. Robot arm consisting of: (1) HEBI X5-4 Series 
Elastic Actuators, (2) black 1.25"$\oslash$ PVC piping, and (3) aluminum 
brackets, supporting a (4) small steel plate.}
\label{fig:ArmPhoto}
\end{figure}

\subsection{Application of iterative machine learning for joint control}
To verify the ability to track trajectories with varying parameters, two 
trajectories were tested: i) a ``slow'' large-range-of-motion trajectory that 
moves both joints through an angle $\pi$, and ii) a ``fast'' 
short-range-of-motion trajectory that moves both joints through an angle 
$\pi/2$, both with a sinusoidal acceleration profile. Application of the 
iterative machine learning algorithm requires addressing the angle-dependent 
dynamics as presented in \cref{sec:ltv} with the joint angles $\theta_1(t)$ and 
$\theta_2(t)$ as the parameters.  Additionally, there is coupling between the 
inputs. Each of the two inputs $u_1(t)$ and $u_2(t)$  have significant 
contribution to the response of both joint angles, so the multiple-input 
multiple-output case was considered as in \cref{sec:mimo}.  The kernel function 
for each transfer function was selected as the squared-exponential kernel with 
automatic relevance determination \cite{Rasmussen_06},
\begin{align}
	\krn_{ij}(\mathbf{x}_1,\mathbf{x}_2) = \sigma_f^2\exp \left(
		-\frac{1}{2}
		(\mathbf{x}_1-\mathbf{x}_2)^H
		\Lambda
		(\mathbf{x}_1-\mathbf{x}_2)
	\right)
    \label{Exp_section_kernel}
\end{align}
where $\Lambda = \diag(\ell_1^{-2},\, \ell_2^{-2},\,\ell_3^{-2})$ is a matrix of 
length scales for the input variables $\mathbf{x} \equiv \mat{\omega & \theta_1 
													    & 
\theta_2}^T$ and the superscript $H$ denotes the complex conjugate transpose.  

To improve the initial models, a quantized and slowed initial trajectory 
$u_0'(t)$ was used, which is shown in \cref{fig:u0-aug}. This method is similar 
to the augmented training signal used in \cite{Banka17IML}, and helps ensure 
that sufficient frequency data is available by sampling the frequency response 
at discrete states along the desired trajectory. 

When training the model, the parameters $\theta_1$ and $\theta_2$ were rounded 
to the nearest $\pi/10 \approx \SI{0.314}{\radian}$ and frequency-response data 
was taken from a \SI{2}{\second} window after each change in the discretized 
joint angles.
After each iteration $k$, the resulting data were thresholded to include only 
the frequency components at which both the input and output magnitudes were at 
least 50\% of the maximum magnitudes in the window, which selects the most 
relevant data and reduces the number of training points to save computational 
cost.  The additional data from iteration $k$ was combined with the data from 
iterations $1$ through $k-1$ to obtain the training set. Then hyperparameters 
for each output kernel $\krn'_i$ were obtained by maximizing the marginal 
likelihood for the training data as described in \cite{Rasmussen_06}.  After 
learning the hyperparameters, the model was estimated at each unique joint angle 
combination in $y_k(t)$, e.g.  as in \cref{fig:model-bode_1,fig:model-bode_2}, 
and the correction in \cref{eq:ltv-update-f} was applied at all time points 
matching those parameters, resulting in an updated input signal $u_{k+1}(t)$.  
When computing the update, the iteration gains were computed as $60\%$ of the 
maximum given in \cref{eq:mimo-rho-bound} with the uncertainty bounds selected 
as
\begin{align}
	\Delta_a = \Delta_b = 2\sqrt{\mathbb{V}[f(\omega)]}.
\end{align} 

\begin{figure}[tb] 
	\begin{center}
		\includegraphics[width=\columnwidth]{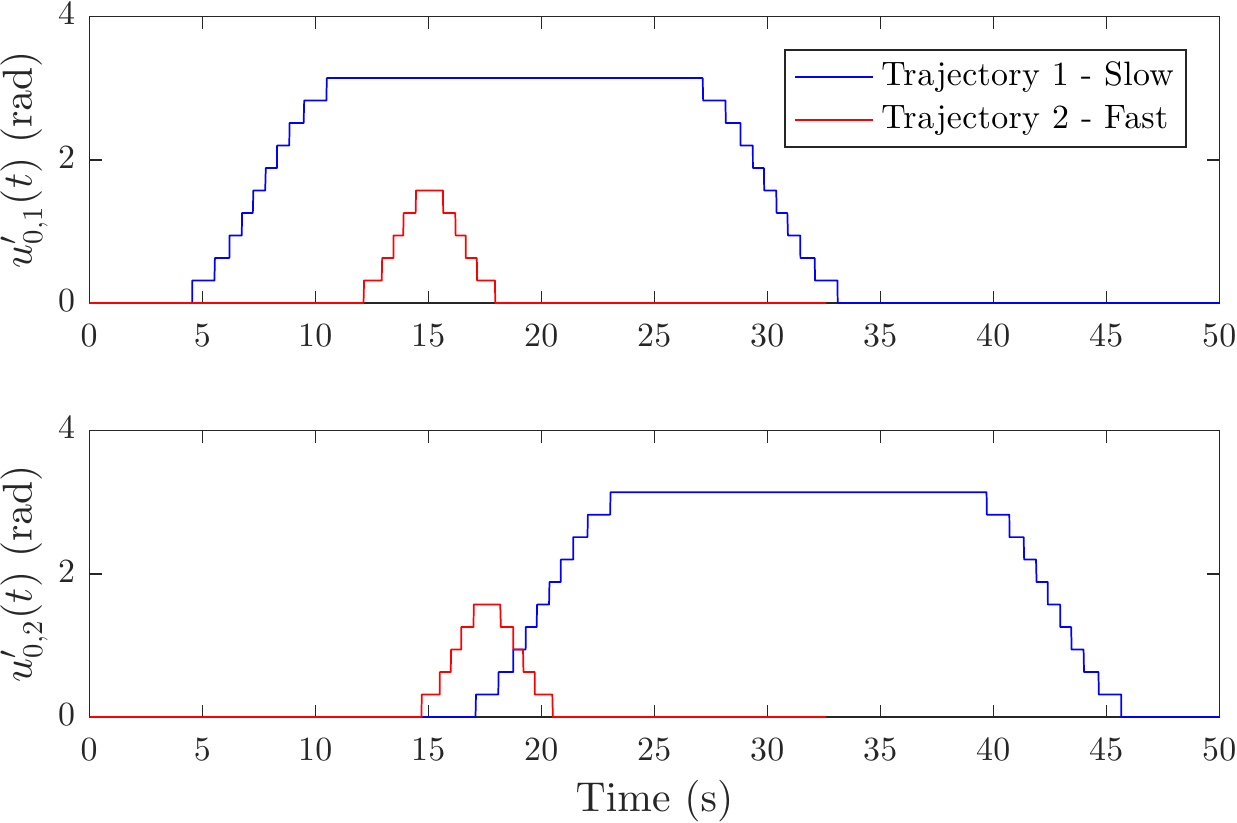}
	\end{center}
	\caption{Augmented initial trajectory $u_0'(t)$}
	\label{fig:u0-aug}
\end{figure}

\begin{figure}[tb] 
	\begin{center}
		\includegraphics[width=\columnwidth]{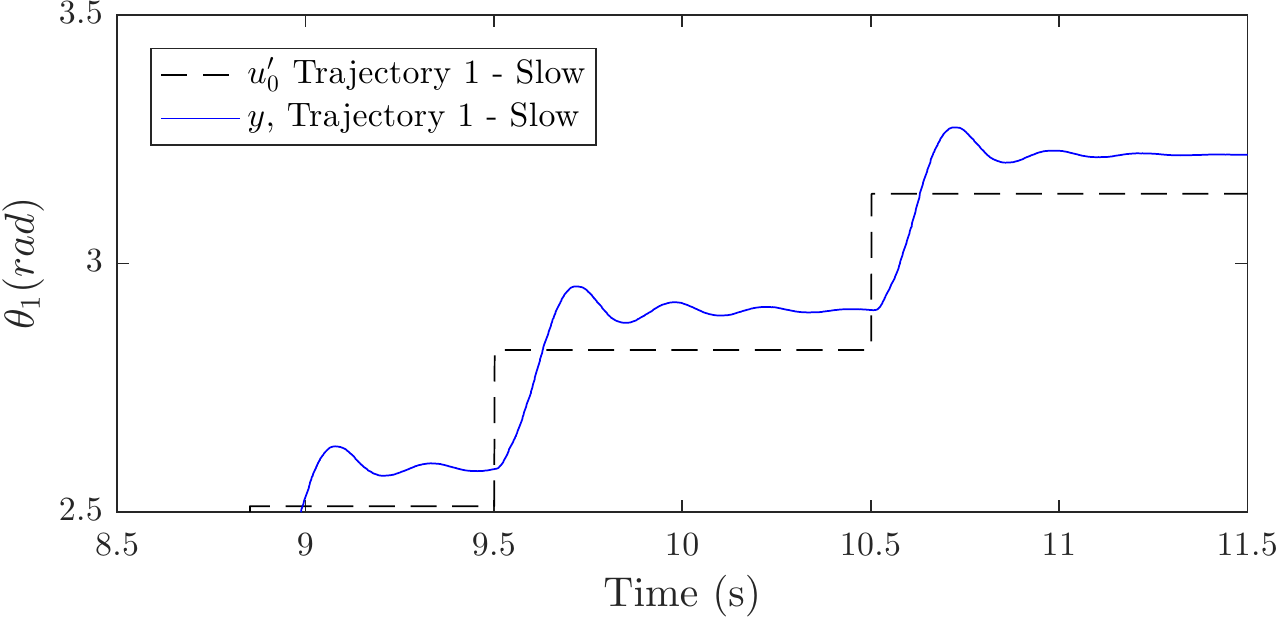}
	\end{center}
	\caption{Step response of $\theta_1$ from the augmented initial trajectory 1, $u_0'(t)$, shown in \cref{fig:u0-aug}.}
	\label{fig:stepResponse}
\end{figure}





\subsection{Tracking results}

\begin{figure}[tb] 
	\begin{center}
		\includegraphics[width=\columnwidth]{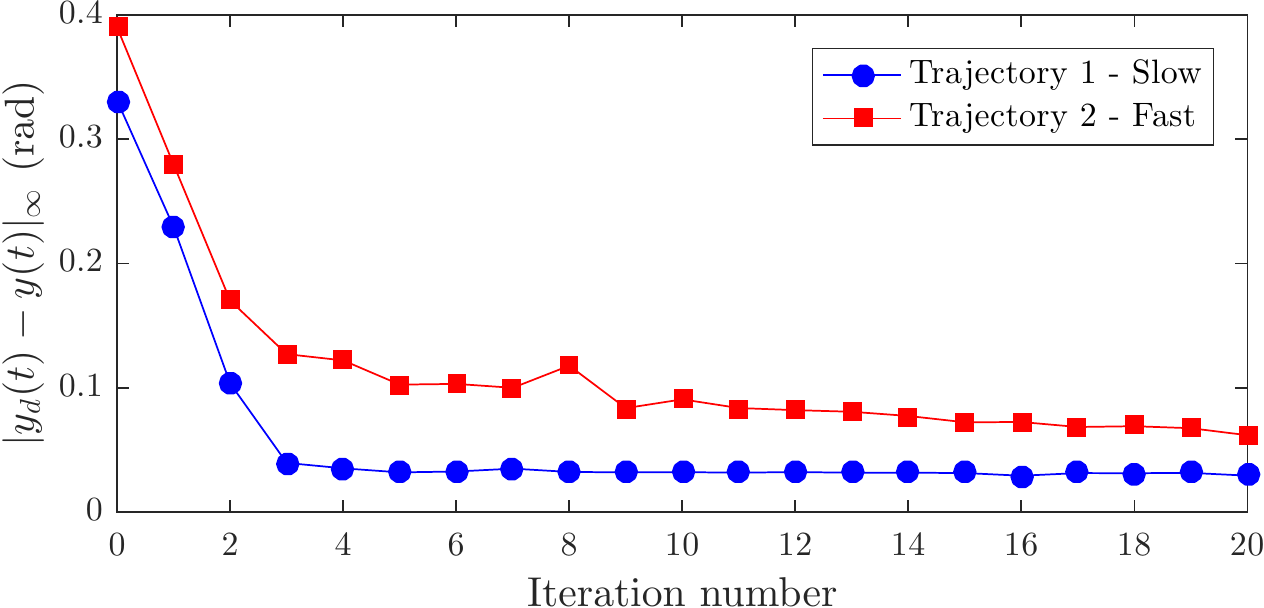}
	\end{center}
	\caption{Convergence of the iterative machine learning procedure for the 
	series-elastic robot arm.}
	\label{fig:convergence}
\end{figure} 

The iterative gains derived in \cref{eq:mimo-rho-bound} and the multiple-input 
iterative machine learning procedure were successfully used to control the joint 
angles of the series elastic robot, as shown by the convergence behavior in 
\cref{fig:convergence}, with a final full range of motion error of \SI{0.029}{\radian} and \SI{0.062}{\radian}, or 
\SI{1}{\percent} and \SI{4}{\percent}, for trajectories 1 and 2, respectively. This is a 91\% and a 84\% reduction in error for trajectories 1 and 2, respectively, compared to using the DC gain $G_0$ alone.  Convergence was achieved after five 
iterations for the slow trajectory, suggesting that the bounds on the iteration gain were not overly 
restrictive. Convergence for the faster trajectory is slower and seems to continue past 20 iterations. Such slower convergence indicates larger modeling errors at the higher frequencies. Also, the time variations in the system, which are neglected in the modeling, can be more significant, which also impacts convergence. 
The initial error using $u_0(t) = y_d(t)/G_0$ and the error for the final 
learned input $u_{20}(t)$ are compared in \cref{fig:error0f_1,fig:error0f_2} and 
the corresponding inputs are shown in 
\cref{fig:input-comparison_1,fig:input-comparison_2}. The maximum error values 
for trajectory 1 started at $$\left[ \max|\Delta\theta_1|, \max|\Delta\theta_2| 
\right] _{k=0} = [0.331,0.287] \SI{}{\radian},$$
and reduced to $$\left[\max|\Delta\theta_1|, 
\max|\Delta\theta_2|\right]_{k=20} = [0.029, 0.029]\SI{}{\radian}.$$ 
Similarly, the maximum error values for trajectory 2 started at 
$$\left[\max|\Delta\theta_1|, \max|\Delta\theta_2|\right]_{k=0} = 
[0.390,0.266]\SI{}{\radian},$$
and reduced to $$\left[\max|\Delta\theta_1|, \max|\Delta\theta_2|\right]_{k=20} 
= [0.062,0.044]\SI{}{\radian}.$$
These substantial error reductions between the initial and final output-trajectories, as shown in \cref{fig:OUTPUT-comparison_t1_n0,fig:OUTPUT-comparison_t1_n20,fig:OUTPUT-comparison_t2_n0,fig:OUTPUT-comparison_t2_n20}, 
confirm that the IML method can successfully learn the input to improve the tracking performance of the SEA robot.

\begin{figure}[tb] 
	\begin{center}
		\includegraphics[width=\columnwidth]{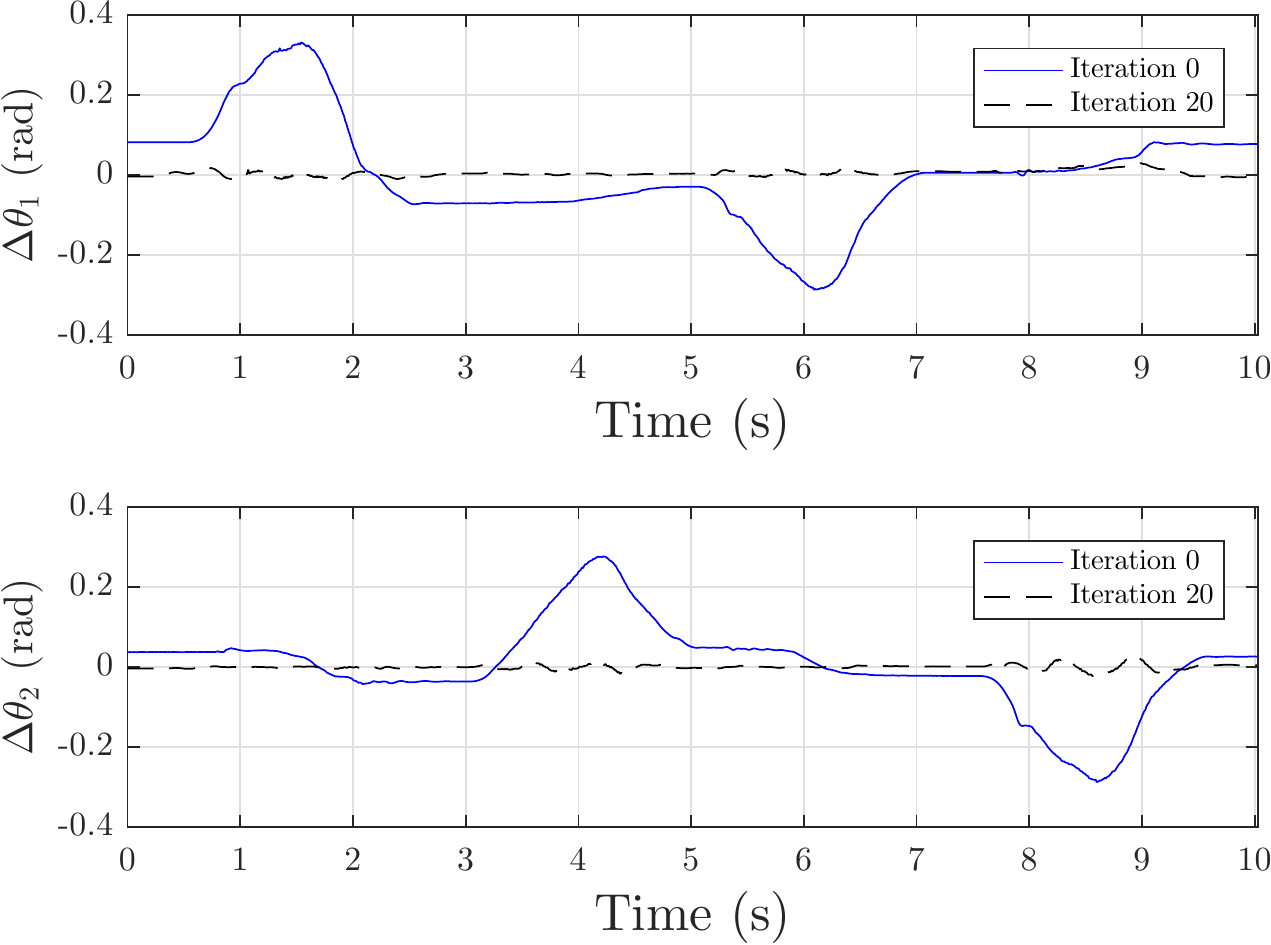}
	\end{center}
	\caption{Trajectory 1. Comparison of the initial ({\color{blue}---}) and final (- - -) 
	error signals on the two joint angles $\theta_1(t)$ and $\theta_2(t)$.}
	\label{fig:error0f_1}
\end{figure}
\begin{figure}[tb] 
	\begin{center}
		\includegraphics[width=\columnwidth]{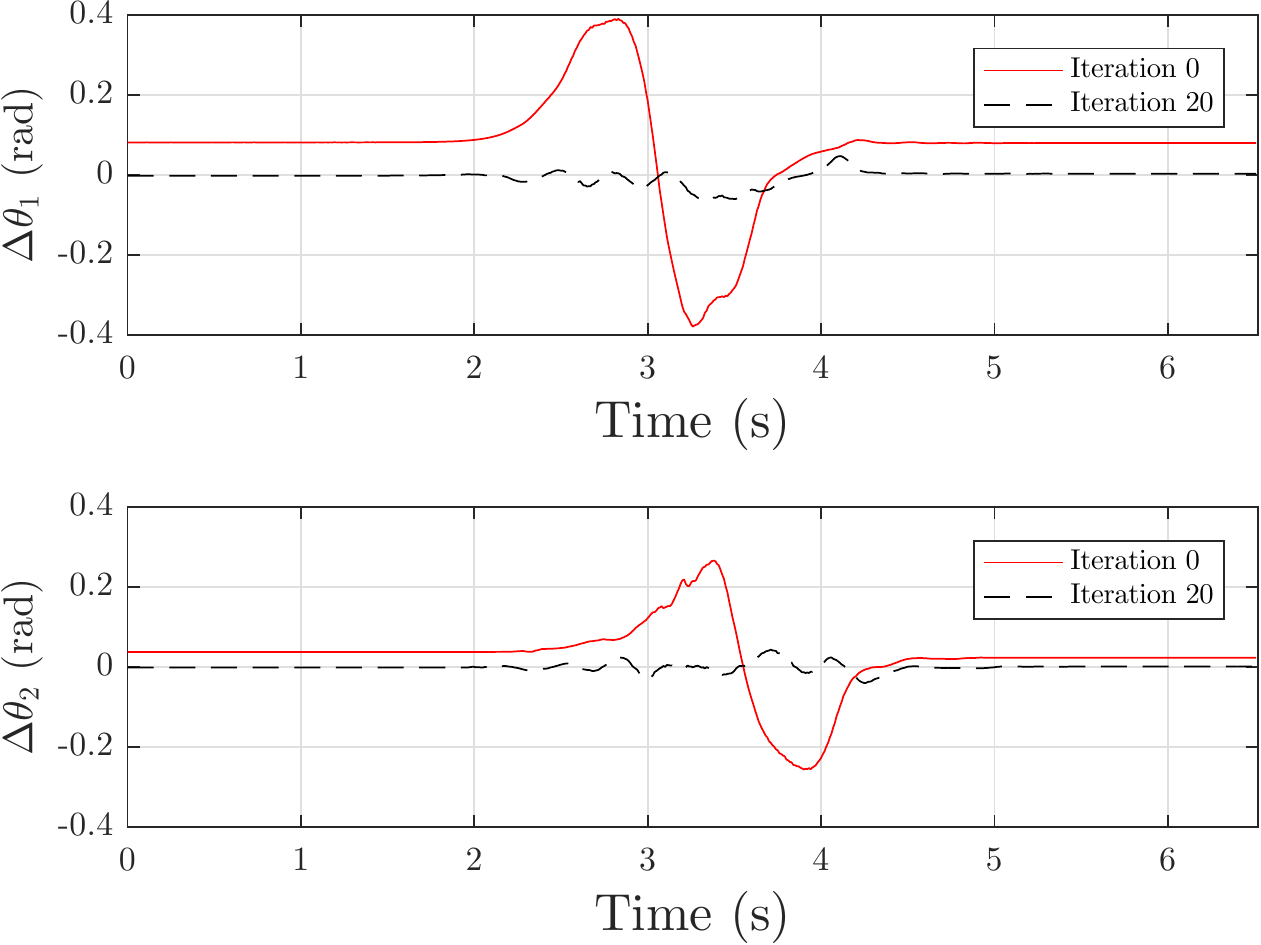}
	\end{center}
	\caption{Trajectory 2. Comparison of the initial ({\color{red}---}) and 
	final (- - -) error signals on the two joint angles $\theta_1(t)$ and 
$\theta_2(t)$.}
	\label{fig:error0f_2}
\end{figure}
\begin{figure}[tb] 
	\begin{center}
		\includegraphics[width=\columnwidth]{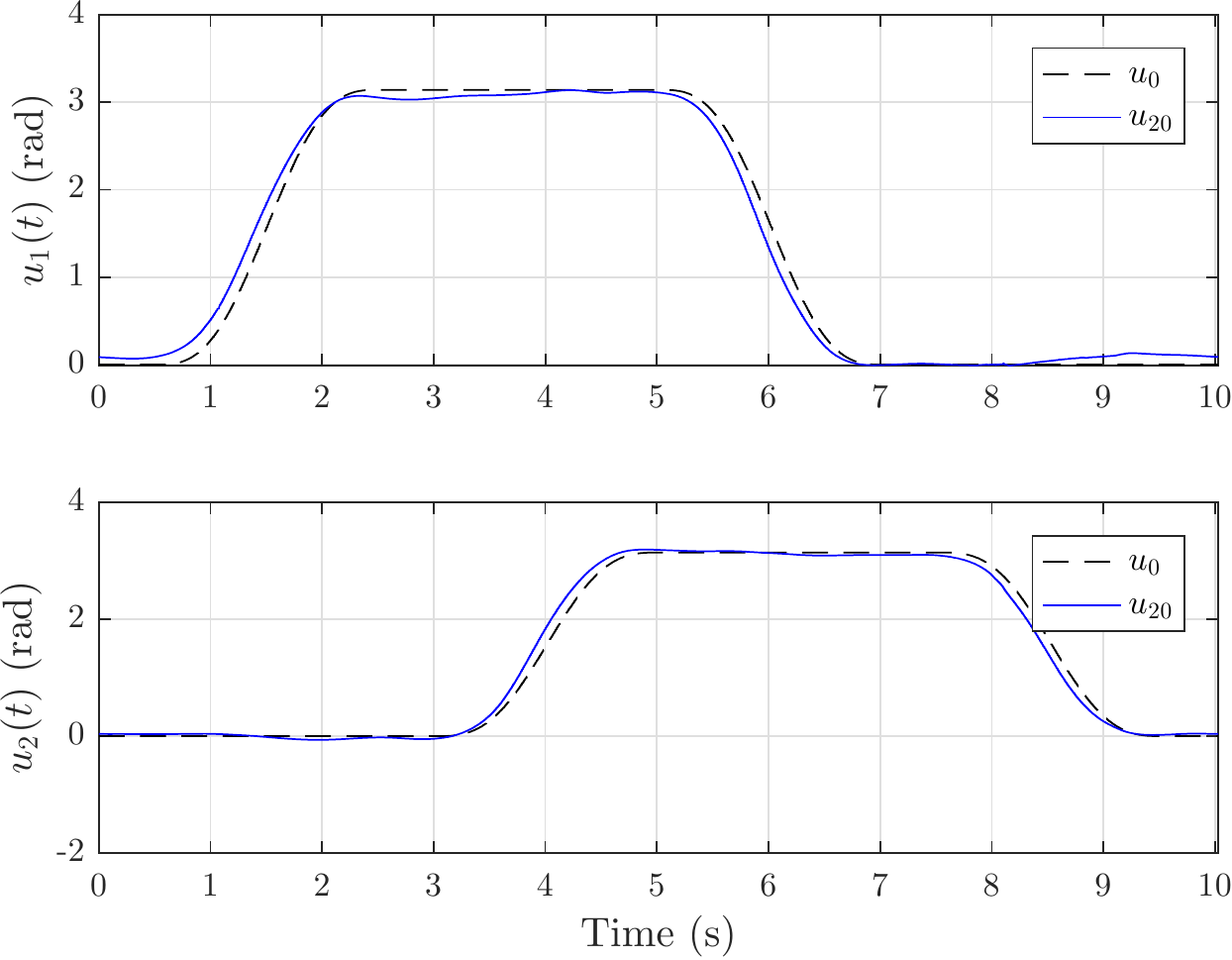}
	\end{center}
	\caption{Trajectory 1. Comparison of the initial input $u_0(t) = y_d/G_0$ (- - -) and the 
	final learned input $u_{20}(t)$ ({\color{blue}---}).}
	\label{fig:input-comparison_1}
\end{figure}
\begin{figure}[tb] 
	\begin{center}
		\includegraphics[width=\columnwidth]{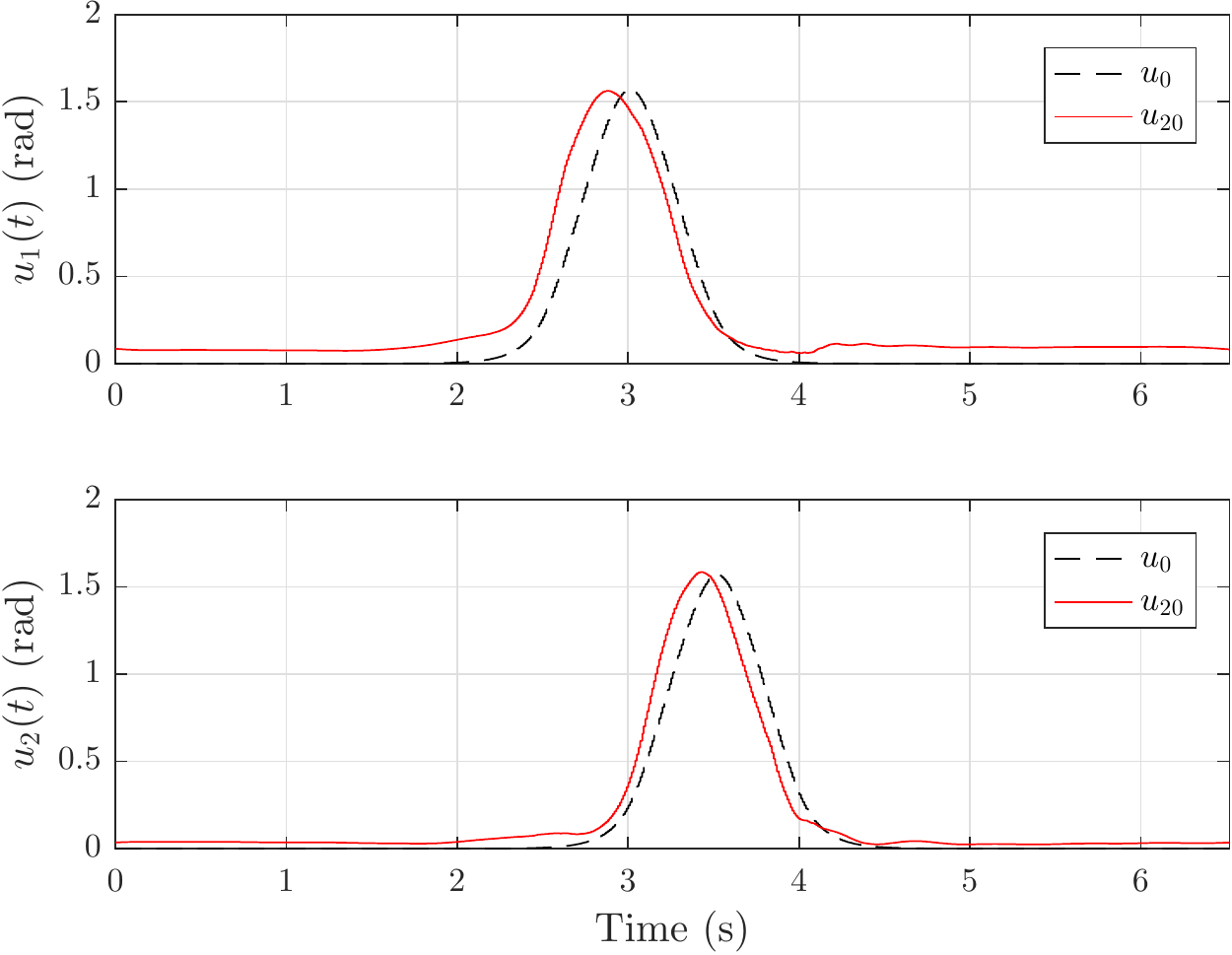}
	\end{center}
	\caption{Trajectory 2. Comparison of the initial input $u_0(t) = y_d/G_0$ (- - -) and the 
	final learned input $u_{20}(t)$ ({\color{red}---}).}
	\label{fig:input-comparison_2}
\end{figure}


\begin{figure} 
	\begin{center}
		\includegraphics[width=\columnwidth]{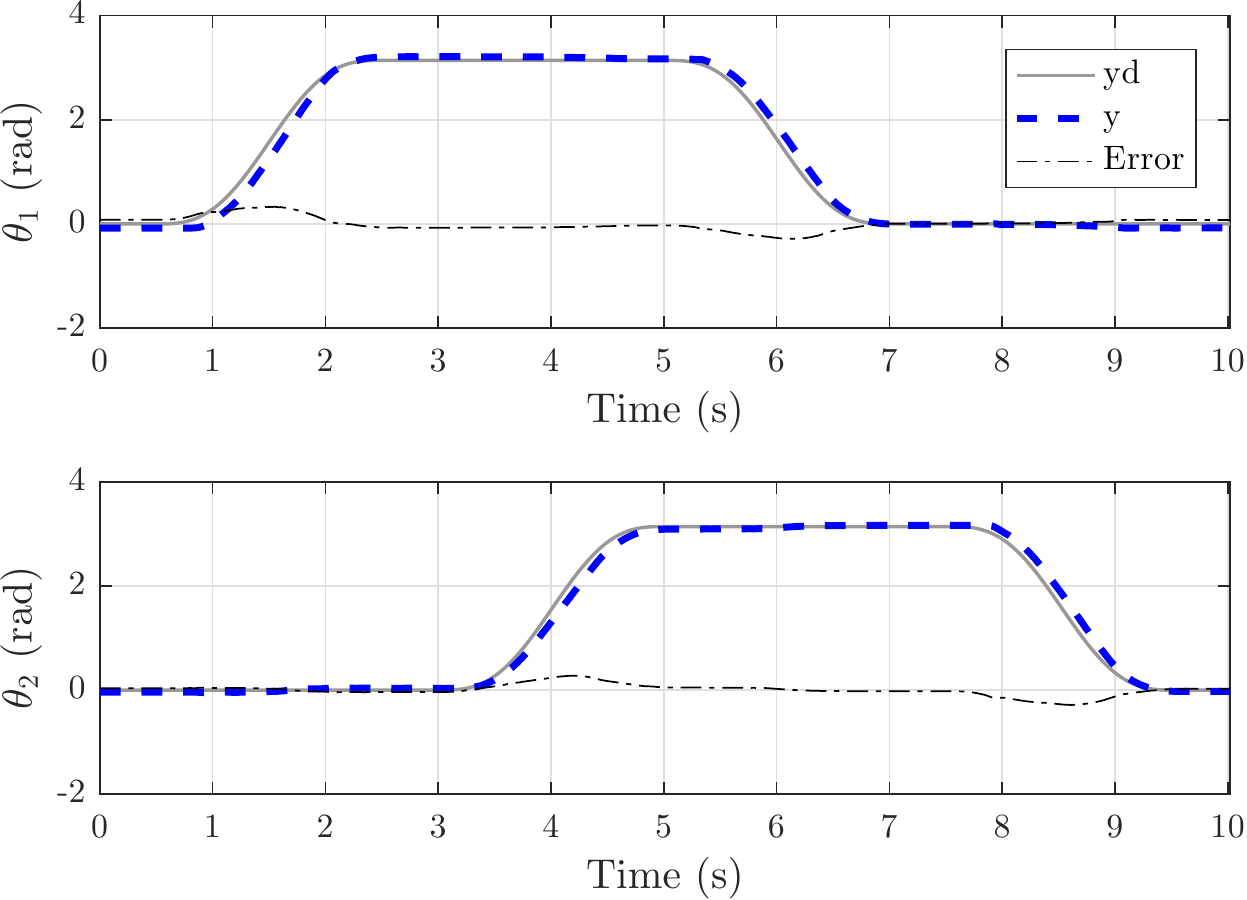}
	\end{center}
	\caption{Output Tracking Trajectory 1. Iteration 0. Comparison of desired trajectory, $y_d$, achieved trajectory, $y$, and error between them.}
	\label{fig:OUTPUT-comparison_t1_n0}
\end{figure}
\begin{figure} 
	\begin{center}
		\includegraphics[width=\columnwidth]{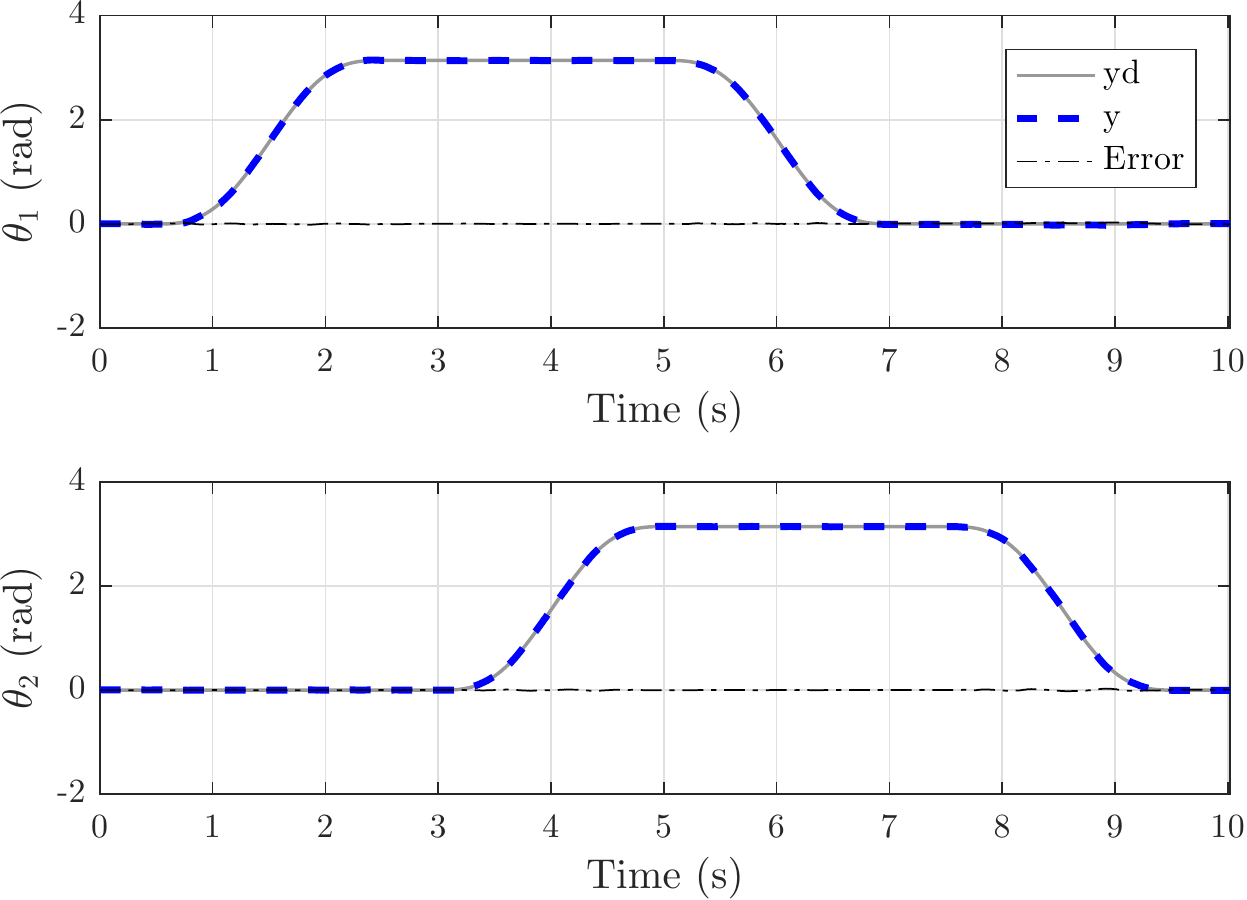}
	\end{center}
		\caption{Output Tracking Trajectory 1. Iteration 20. Comparison of desired trajectory, $y_d$, achieved trajectory, $y$, and error between them.}
	\label{fig:OUTPUT-comparison_t1_n20}
\end{figure}
\begin{figure} 
	\begin{center}
		\includegraphics[width=\columnwidth]{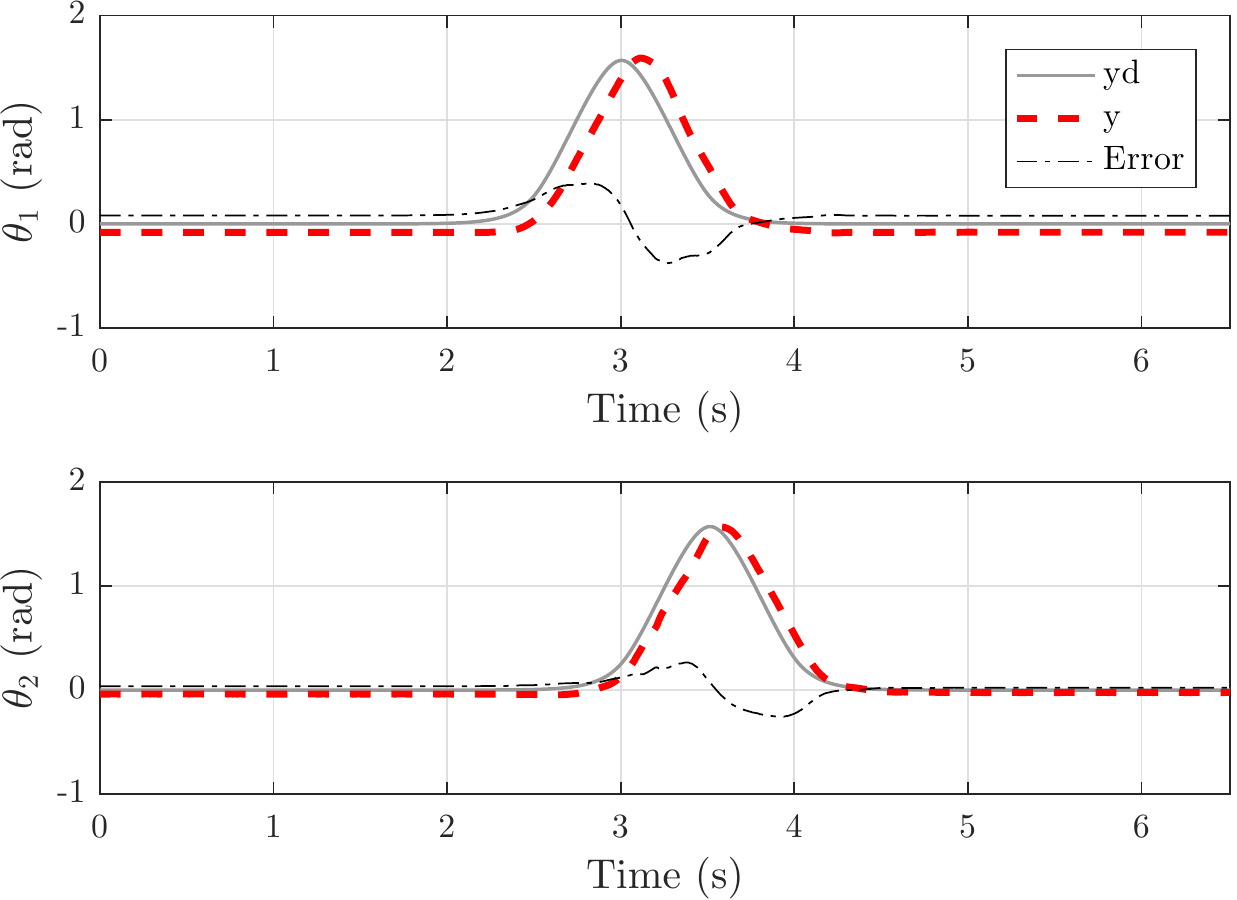}
	\end{center}
		\caption{Output Tracking Trajectory 2. Iteration 0. Comparison of desired trajectory, $y_d$, achieved trajectory, $y$, and error between them.}
	\label{fig:OUTPUT-comparison_t2_n0}
\end{figure}
\begin{figure} 
	\begin{center}
		\includegraphics[width=\columnwidth]{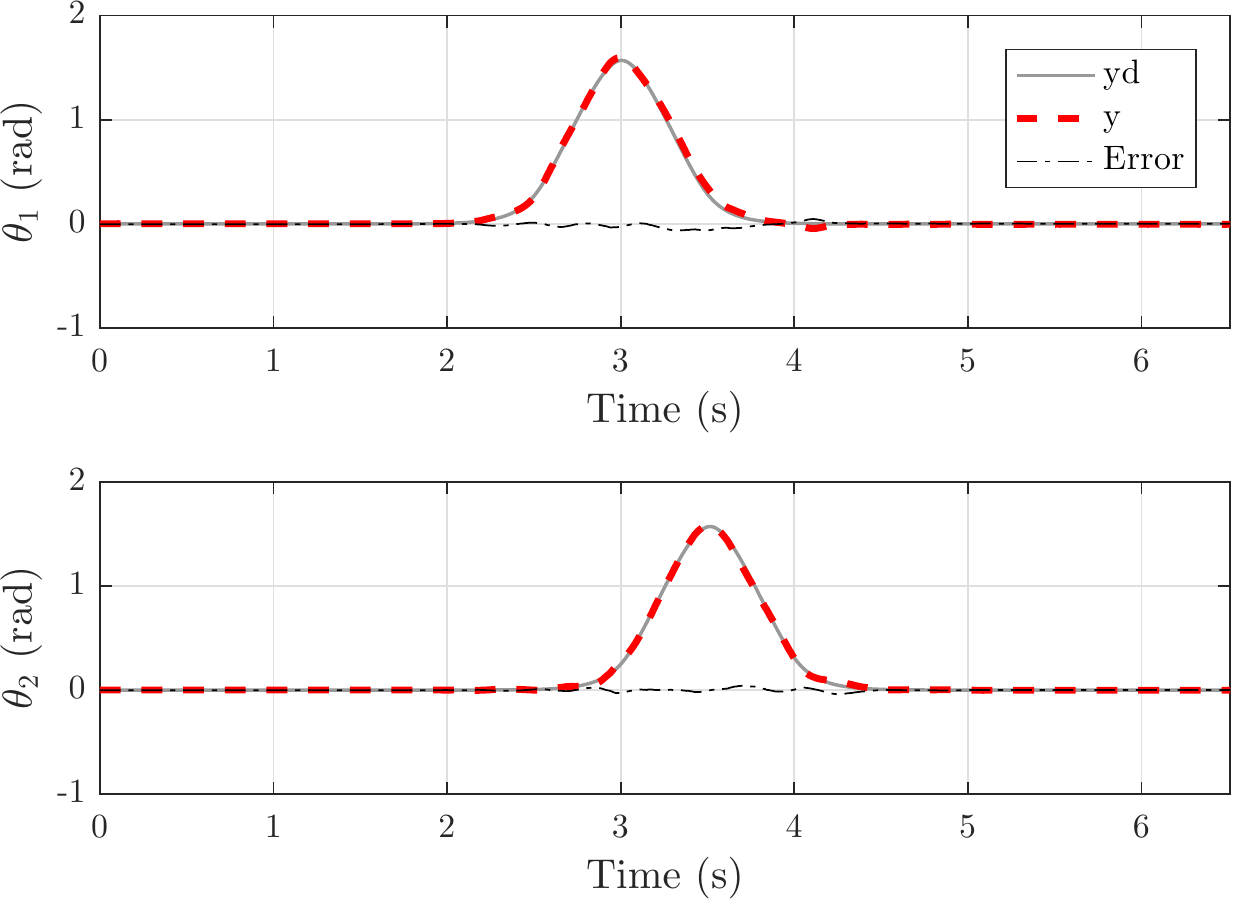}
	\end{center}
		\caption{Output Tracking Trajectory 2. Iteration 20. Comparison of desired trajectory, $y_d$, achieved trajectory, $y$, and error between them.}
	\label{fig:OUTPUT-comparison_t2_n20}
\end{figure}

\begin{figure}[tb] 
	\begin{center}
		\includegraphics[width=0.6\columnwidth]{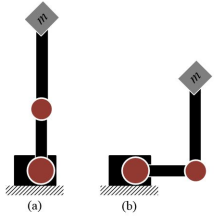}
	\end{center}
	\caption{Robot poses used in generating the bode plots in \cref{fig:model-bode_1,fig:model-bode_2}. (a) $\theta_1$ = $\pi$/2, $\theta_2$ = 0 and (b) $\theta_1$ = 0, $\theta_2$ = $\pi$/2. See \cref{fig:ArmDiagram} for more details.}
    \label{fig:robotPoses}
\end{figure}

\begin{figure} 
	\begin{center}
		\includegraphics[width=\columnwidth]{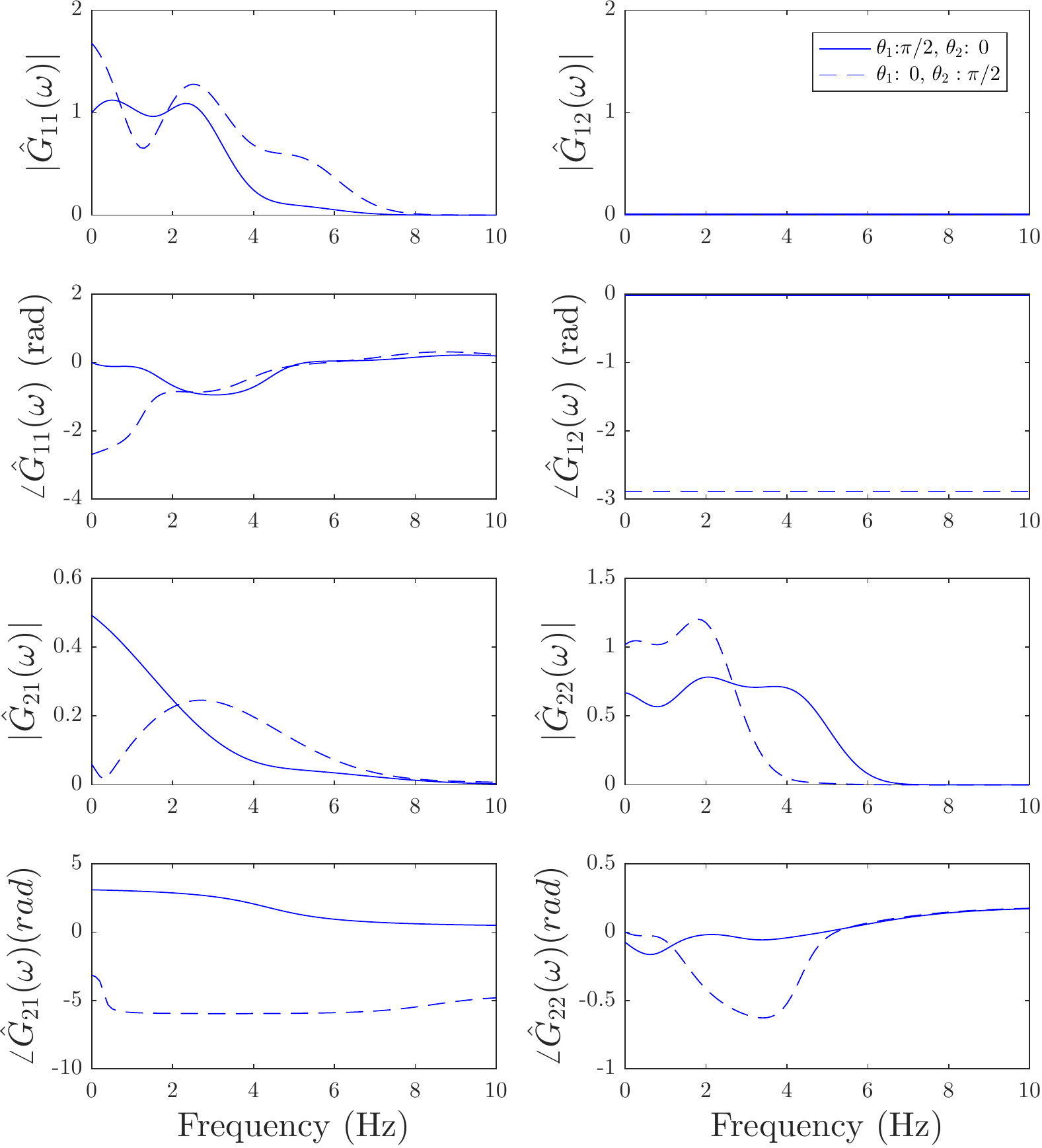}
	\end{center}
	\caption{Trajectory 1. Linearized model estimates of the system after 20 iterations at two different poses (sets of $\theta_1$ and $\theta_2$ positions). See \cref{fig:robotPoses} for poses. The eight plots show magnitude and phase of the MIMO system -- two inputs and two outputs.}
	\label{fig:model-bode_1}
\end{figure}

\begin{figure} 
	\begin{center}
		\includegraphics[width=\columnwidth]{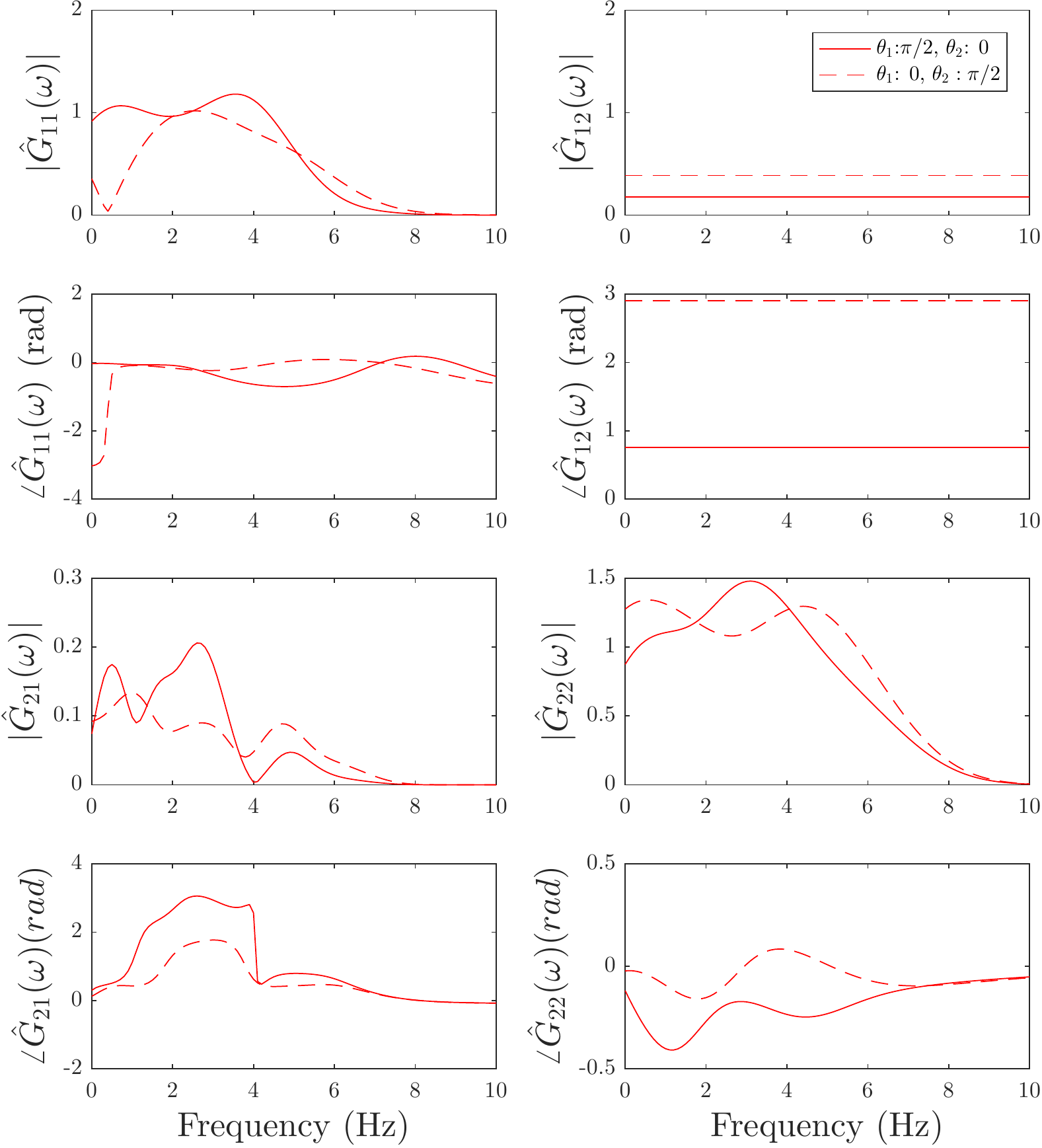}
	\end{center}
	\caption{Trajectory 2. Linearized model estimates of the system after 20 iterations at two different poses (sets of $\theta_1$ and $\theta_2$ positions). See \cref{fig:robotPoses} for poses. The eight plots show magnitude and phase of the MIMO system -- two inputs and two outputs. }
	\label{fig:model-bode_2}
\end{figure}





\section{Conclusions}
An iterative machine learning (IML) method was proposed that 
enables iterative control of control-affine nonlinear multiple-input 
multiple-output systems with unknown dynamics.  Conditions on local convergence were 
presented and the method was tested on a two-link robotic arm driven by series 
elastic actuators. 
The iterative machine learning approach converged to an input with a maximum error in the joint angles of 1\% and 4\% of 
the range of motion for trajectories 1 and 2, respectively.   
Current efforts are aimed at developing global conditions of the 
rate of acceptable trajectory variations to ensure convergence of the proposed 
iterative approach with the localized models. 





\end{document}